\documentclass[11pt]{article}
\usepackage[english]{babel}
\usepackage{graphicx}
\usepackage{amsthm}
\usepackage{proof}
\usepackage{amsmath}
\usepackage{amssymb}
\usepackage[dvips]{color}
\usepackage{enumerate}
\usepackage{url}

\newtheorem{theorem}{Theorem}
\newtheorem{corollary}{Corollary}
\newtheorem{lemma}{Lemma}
\newtheorem{exa}[theorem]{Example}
\newenvironment{example}{\begin{exa} \rm}{\end{exa}}

\newcommand{\Ag}{\textit{Ag} }

\newcounter{symbol}
\newcommand{\indexsyma}[1]%
{\stepcounter{symbol}\index{zzz1 \thesymbol @\protect#1}}
\newcommand{\indexsymb}[1]%
{\stepcounter{symbol}\index{zzz2 \thesymbol @\protect#1}}
\newcommand{\indexsymc}[1]%
{\stepcounter{symbol}\index{zzz3 \thesymbol @\protect#1}}
\newcommand{\indexsymd}[1]%
{\stepcounter{symbol}\index{zzz4 \thesymbol @\protect#1}}
\newcommand{\indexsyme}[1]%
{\stepcounter{symbol}\index{zzz5 \thesymbol @\protect#1}}

\newcommand{\bfe}[1]{\begin{bfseries}\emph{#1}\end{bfseries}\index{#1}}

\newcommand{\szkew}[1]{\relax \setbox0=\hbox{\kern -24pt $\displaystyle#1$\kern 0pt }%
\box0}
{\catcode`\@=11 \global\let\ifjusthvtest@=\iffalse}

\newcounter{oldmycaption}

\newcommand{\Lang}{\mathfrak{L}}

\renewcommand{\L}{L}

\title{Common Knowledge in Email Exchanges}

\author{
Floor Sietsma$^*$ and
Krzysztof R. Apt\footnote{Centre for
    Mathematics and Computer Science (CWI), Science Park 123, 1098 XG
    Amsterdam, the Netherlands,
    and University of Amsterdam
}
}
\date{}
\begin{document}

\maketitle

\begin{abstract}
  We consider a framework in which a group of agents communicates by
  means of emails, with the possibility of replies, forwards and blind
  carbon copies (BCC).  We study the epistemic consequences of such
  email exchanges by introducing an appropriate epistemic language and
  semantics.  This allows us to find out what agents learn from the
  emails they receive and to determine when a group of agents acquires
  common knowledge of the fact that an email was sent.
  We also show that in our framework from the epistemic point of view the BCC feature
  of emails cannot be simulated using messages without BCC recipients.
\end{abstract}

\section{Introduction}

\subsection{Motivation}

Email is by now a prevalent form of communication.  From the point of
view of distributed programming it looks superficially as an instance
of multicasting ---one agent sends a message to a group of agents.
However, such features as forwarding and the \emph{blind carbon copy}
(BCC) make it a more complex form of communication.

The reason is that each email implicitly carries epistemic information
concerning (among others) common knowledge of the group involved in it
of the fact that it was sent.  As a result forwarding leads to nested common
knowledge and typically involves at each level different groups of
agents.  In turn, the BCC feature results in different information
gain by the regular recipients and the BCC recipients.
In fact, in Section~\ref{sec:bcc} we show that the BCC feature
is new from the epistemic point of view.

To be more specific, suppose that an agent $i$ forwards a message $m$
to a group $G$.  Then the group $G \cup \{i\}$ acquires (among others)
common knowledge of the fact that the group $A$ consisting of the
sender and the receivers of $m$ has common knowledge of $m$.
Next, suppose that an agent $i$ sends a message $m$ to a
group $G$ with a BCC to a group $B$. Then 
the group $G \cup \{i\}$ acquires common knowledge of $m$,
while each member of $B$ separately 
acquires with the sender of $m$ common knowledge of the fact that
the group $G \cup \{i\}$ acquires common knowledge of $m$.

Combining forward and BCC we can realize epistemic formulas $C_{A_1}...C_{A_k} m$, where $C_A$ stands for `the group $A$ has common
knowledge of', of arbitrary depth. Further, this combination can lead
to a(n undesired) situation in which a BCC recipient of an email
reveals his status to others by using the \emph{reply-all} feature.
In general, a chain of forwards of arbitrary length can reveal to a
group of agents that an agent was a BCC recipient of the original
email.  We conclude that the email exchanges, as studied here, are
essentially different from multicasting.

Epistemic consequences of email exchanges are
occasionally raised by researchers in various contexts.
For instance, the author of \cite{Bab90} mentions `some issues of email ethics' 
by discussing a case of an email discussion in which some researchers
were not included (and hence could not build upon the reported results).

Then consider the following recent quotation from a blog
in which the writers call for a boycott of a journal XYZ: ``We are
doing our best to make the misconduct of the Editors-in-Chief a matter
of common knowledge within the [...] community in the hope that
everyone will consider whatever actions may be appropriate for them to
adopt in any future associations with XYZ''.

So when studying email exchanges a natural question arises: what are
their knowledge-theoretic consequences? To put it more informally:
after an email exchange took place, who knows what?  Motivated by the
above blog entry we can also ask: can sending emails to more and more new
recipients ever create common knowledge? (Our Main Theorem
shows that the answer is ``No.'')

To be more specific consider the following example to which we shall return later.
\begin{example} \label{exa:1}
Assume the following email exchange involving four people,
Alma, Bob, Clare and Daniel:

\begin{itemize}
\item Alma and Daniel got an email from Clare,

\item Alma forwarded it to Bob,

\item Bob forwarded Alma's email to Clare and Daniel with a BCC to Alma,

\item Alma forwarded the last email to Clare and Daniel with a BCC to Bob.

\end{itemize}
It is natural to ask for example what Alma has actually learned from Bob's email.
Also, do all four people involved in this exchange have common knowledge of the original email by Clare?
\end{example}

To answer such questions we study email exchanges focusing on relevant features that we
encounter in most email systems.
More specifically, we study the following form of email communication:
\begin{itemize}

\item each email has a sender, a non-empty set of regular recipients and a
(possibly empty) set of blind carbon copy (BCC)
  recipients. Each recipient receives a copy of the message and is only
aware of the regular recipients and not of the BCC recipients (except himself),

\item in the case of a reply to or a forward of a message, the \emph{unaltered} original message is included, 

\item in a reply or a forward, one can append new information to the original message one replies to or forwards.
  
\end{itemize}

To formalize agents' knowledge resulting from an email exchange we
introduce an appropriate epistemic language and the corresponding
semantics.  The resulting model of email communication differs from
the ones that were studied in other papers in which only limited
aspects of emails have been considered. These papers are discussed
below. In our setup the communication is synchronous. While this is a
simplification we find that it is natural to clarify email
communication in such a setting first before considering various
alternatives.  In the last section we address this point further when
suggesting further research.

\subsection{Contributions and plan of the paper}

To study the relevant features of email communication we introduce in
the next section a carefully chosen language describing emails.  We
make a distinction between a \emph{message}, which is sent to a public
recipient list, and an \emph{email}, which consists of a message and a
set of BCC recipients.  This distinction is relevant because a
\emph{forward} email contains an earlier message, without the list of
BCC recipients.  We also introduce the notion of a legal state that
imposes a natural restriction on the considered sets of emails by
stipulating an ordering of the emails. For example, an email
needs to precede any forward of it.

To reason about the knowledge of the agents after an email exchange
has taken place we introduce in Section \ref{sec:epistemic} an
appropriate epistemic language.  Its semantics takes into account the
uncertainty of the recipients of an email about its set of BCC
recipients. This semantics allows us to evaluate
epistemic formulas in legal states, in particular the formulas that
characterize the full knowledge-theoretic effect of an email.

Apart from factual information each email also carries epistemic
information. In Section \ref{sec:EI} we characterize the latter. It
allows us to clarify which groups of agents acquire common knowledge
as a result of an email and what is the resulting information gain
for each agent.

In Section \ref{sec:common} we present the main result of the paper,
that clarifies when a group of agents can acquire
common knowledge of the formula expressing
the fact that an email has been sent.  This characterization in particular 
sheds light on the epistemic consequences of BCC.
The proof is given in Section \ref{sec:proof}.

Then in Section \ref{sec:bcc} we show that in our framework 
BCC cannot be simulated using messages without BCC recipients.
Finally, in Section \ref{sec:dis}
we provide a characterization of legal states in terms of
properly terminating email exchanges. 

\subsection{Related work}

The study of the epistemic effects of communication in distributed
systems originated in the eighties and led to the seminal book
\cite{FHMV_RAK}. The relevant literature, including
\cite{chandy_processes_1986}, deals with the 
communication forms studied within the context of distributed computing, 
notably asynchronous send.

One of the main issues studied in these frameworks has been the
analysis of the conditions that are necessary for acquiring common
knowledge.  In particular, \cite{HM90} showed that common knowledge
cannot be attained in the systems in which the message delivery is not
guaranteed.  More recently this problem was investigated in
\cite{BM10} for synchronous systems with known bounds on message
transmission in which processes share a global clock.  The authors
extended the causality relation of \cite{Lam78} between messages in
distributed systems to synchronous systems with known bounds on
message transmission and proved that in such systems a so-called
pivotal event is needed in order to obtain common knowledge.  This in
particular generalizes the previous result of
\cite{chandy_processes_1986} concerning acquisition of common
knowledge in distributed systems with synchronous communication.

The epistemic effects of other forms of communication were studied in
numerous papers.  In particular, in \cite{pacpar07} the communicative
acts are assumed to consist of an agent~$j$ `reading' an arbitrary
propositional formula from another agent~$i$.  The idea of an
epistemic content of an email is implicitly present in \cite{PR03},
where a formal model is proposed that formalizes how communication
changes the knowledge of a recipient of the message.  

In \cite{vBvEK06} a dynamic epistemic logic modelling effects of
communication and change is introduced and extensively studied.
Further, in \cite{WSE10} an epistemic logic was proposed to reason
about information flow w.r.t.~underlying communication channels.
\cite{Pac10} surveys these and related approaches and discusses the
used epistemic, dynamic epistemic and doxastic logics.

Most related to the work here reported are the following two
references.  \cite{AWZ09} studied knowledge and common knowledge in a
set up in which the agents send and forward propositional formulas in
a social network.  However, the forward did not include the original
message and the BCC feature was absent.  More recently, in
\cite{SvE11} explicit messages are introduced in a dynamic epistemic
logic to analyze a similar setting, though BCC was simulated as
discussed in Section \ref{sec:bcc}.  In both papers it is assumed that
the number of messages is finite. In contrast, in the setting of this
paper the forward includes the original message, which results
directly in an infinite number of messages and emails.  Finally, let
us mention that the concept of forwarding is occasionally mentioned in
the context of distributed computing, see, e.g., \cite{EC09}.

\section{Preliminaries}
\label{sec:preliminaries}

\subsection{Messages}

In this section we define the notion of a message. In the next section
we introduce emails as a simple extension of the messages.  We assume
a non-empty and finite set of \bfe{agents} $\Ag = \{1, ..., n\}$ and a
set of \bfe{notes}.  Each note is an abstraction of the contents of
the message or an email.

We make a number of assumptions. Firstly, we assume that initially
each agent $i$ has a set of notes $L_i$ he knows.  He does know which
notes belong to the other agents and does not know the overall set of
notes.  Furthermore, we assume that the agents only exchange messages
about the notes.  We also assume that an agent can send a message to
other agents containing a note only if he holds it initially or has
learnt it through a message he received earlier.

This minimal set up precludes the possibility that the agents can use messages to
implement some agreed in advance protocol, such as that sending two
specific notes by an agent would reveal that he has some specific
knowledge.  It allows us to focus instead on the epistemic information
caused \emph{directly} by the structure of the messages and emails.

Of course in reality emails may contain propositional or epistemic
information which affects knowledge of the agents at a deeper level
than modelled here by means of abstract notes.  To reason about notes
containing such information one could add on the top of our framework
an appropriate logic.  If every note $n$ contains some formula
$\varphi_n$, then one could just add the implications $n \rightarrow
\varphi_n$ to this logic to ensure that every agent who knows the note
$n$ also knows the formula $\varphi_n$.

We inductively define \bfe{messages} as follows, where we assume that $G \neq \emptyset$:
\begin{itemize}
\item $m := s(i,l,G)$; the message containing note $l$, sent by agent $i$ to the group $G$,
\item $m := f(i,l.m',G)$; the forwarding by agent $i$ of the message $m'$ with added note $l$, sent to the group $G$.
\end{itemize}

So the agents can send a message with a note or forward a message with
a new note appended, where the latter covers the possibility of a 
reply or a reply-all. Appending such a new note to a forwarded message
is a natural feature present in most email systems.
To allow for the possibility of sending a forward without appending 
a new note, we can assume there exists a note \textbf{true} that is held
by all agents and identify $\textbf{true}.m$ with $m$.

If $m$ is a message, then we denote by $S(m)$ and $R(m)$,
respectively, the singleton set consisting of the agent sending $m$
and the group of agents receiving $m$. So for the above messages $m$
we have $S(m) = \{i\}$ and $R(m) = G$.  We do allow that $S(m)
\subseteq R(m)$, i.e., that one sends a message to oneself.

Special forms of the forward messages can be used to model reply
messages.  Given $f(i,l.m,G)$ with $i \in R(m)$, using $G = S(m)$ we obtain the
customary \emph{reply} message and using $G = S(m) \cup R(m)$ we obtain the
customary \emph{reply-all} message.  
(In the customary email systems there is syntactic difference between
a forward and a reply to these two groups of agents, but the effect of
both messages is exactly the same, so we ignore this difference.)  
In the examples we write $s(i,l,j)$ instead of $s(i,l,\{j\})$, etc.

\subsection{Emails}
\label{subsec:emails}

An interesting feature of most email systems is that of the blind carbon copy (BCC).
We  study here the epistemic effects of sending an email with BCC recipients
and will now include this feature in our presentation.

In the previous subsection we defined messages that have a sender and a 
group of recipients. Now we define the notion of an email which
allows the additional possibility of sending a BCC of a message.  The
BCC recipients are not listed in the list of recipients, therefore we
have not included them in the definition of a message. Formally, by an
\bfe{email} we mean a construct of the form $m_B$, where $m$ is a
message and $B \subseteq \Ag$ is a possibly empty set of BCC recipients.
Given a message $m$ we call each email $m_B$ a \bfe{full version} of
$m$.

An email $m_B$ is delivered to the regular recipients, i.e., to the set $R(m)$
and to the set $B$ of the BCC recipients. Each of them receives the message
$m$. Only the sender of $m_B$, i.e., agent $i$, where $S(m) = \{i\}$,
knows the set $B$. Each agent $i \in B$ only knows that the set $B$
contains at least him.

Since the set of the BCC recipients is `secret', it does not appear in
a forward. That is, the forward of an email $m_B$ with added note $l$
is the message $f(i, l.m, G)$ or an email $f(i, l.m, G)_C$, in which
$B$ \emph{is not mentioned}.  This is consistent with the way BCC is
handled in most email systems, such as \texttt{gmail} or email systems
based on the \texttt{postfix} mail server.  However, this forward may
be sent not only by a sender or a regular recipient of $m_B$, but also
by a BCC recipient.  Clearly, the fact that an agent was a BCC
recipient of an email is revealed at the moment he forwards the
message.

A natural question arises: what if someone is both a regular recipient
and a BCC recipient of an email? In this case, no one (not even this
BCC recipient himself) would ever notice that this recipient was also
a BCC recipient since everyone can explain his knowledge of the
message by the fact that he was a regular recipient. Only the sender
of the message would know that this agent was also a BCC recipient.
This fact does not change anything and hence 
we assume that for every email $m_B$ we have $(S(m) \cup R(m)) \cap B = \emptyset$.

\begin{example}
Using the just introduced language we can formalize the story from
Example \ref{exa:1} as follows, where we abbreviate Alma to $a$, etc.:
\begin{itemize}
\item Alma and Daniel got an email from Clare:

$e_0 := m_{\emptyset}$, where $m := s(c, l, \{a,d\})$,

\item Alma forwarded it to Bob:

$e_1 := m'_{\emptyset}$, where $m' := f(a, m, b)$, 

\item Bob forwarded Alma's email to Clare and Daniel with a BCC to Alma:

$e_2 := m''_{\{a\}}$, where $m'' := f(b, m', \{c, d\})$,  

\item Alma forwarded the last email to Clare and Daniel with a BCC to Bob:

$e_3 := f(a, m'', \{c, d\})_{\{b\}}$.

\end{itemize}
\end{example}

\subsection{Legal states}
\label{subsec:legalstates}

Our goal is to analyze knowledge of agents after some email exchange took place.
To this end we need to define a possible collection of sent
emails.

First of all, we shall assume that every message is used only once.
In other words, for each message $m$ there is at most one full version
of $m$, i.e., an email of the form $m_B$.  The rationale behind this
decision is that a sender of $m_B$ and $m_{B'}$ might equally well
send a single email $m_{B \cup B'}$.  This assumption can be
summarized as a statement that the agents do not have `second
thoughts' about the recipients of their emails. It also simplifies
subsequent considerations.

In this work we have decided not to impose a total ordering on the
emails in our model, for example by giving each email a time stamp.
This makes the model a lot simpler. Also, many interesting questions
can be answered without imposing such a total ordering.
For example, we can investigate the existence of common knowledge
in a group of agents after an email exchange perfectly well without
knowing the exact order of the emails that were sent.

However, we have to impose some ordering on the sets of emails.
For example, we need to make sure that the agents only send information
they actually know. Moreover, a forward can only be sent after the
original email was sent.
We will introduce the minimal partial ordering that takes care of such issues.

First, we define by structural induction the \bfe{factual information} $FI(m)$ contained
in a message $m$ as follows:
\begin{eqnarray*}
FI(s(i,l,G))   &:=& \{l\},\\
FI(f(i,l.m,G)) &:=& FI(m) \cup \{l\}.
\end{eqnarray*}
Informally, the factual information is the set of notes which occur somewhere
in the message, including those occurring in forwarded messages.

We will use the concept of a \bfe{state} to model the effect of an email exchange.
A state $s = (E, \L)$ is a tuple consisting of a finite set $E$ of emails that were sent
and a sequence $\L = (L_1, ..., L_n)$ of sets of notes for all agents.
The idea of these sets is that each agent $i$ initially holds the notes in $L_i$.
We use $E_s$ and $\L_s$ to denote the corresponding elements of a state $s$,
and $L_1, ..., L_n$ to denote the elements of $\L$.

We say that a state $s = (E,\L)$ is \bfe{legal} if a strict
partial ordering (in short, an spo) $\prec$ on $E$ exists that satisfies the following conditions:

\begin{enumerate}[L.1:]
\item\label{leg1} for each email $f(i,l.m,G)_B \in E$ an email $m_C \in E$ exists such that 
$m_C \prec f(i,l.m,G)_B$ and $i \in S(m) \cup R(m) \cup C$,

\item\label{leg2} for each email $s(i,l,G)_B \in E$, where $l \not\in L_i$,
 an email $m_C \in E$ exists such that $m_C \prec s(i,l,G)_B$, $i \in R(m) \cup C$ and $l \in FI(m)$,

\item\label{leg3} for each email $f(i,l.m',G)_B \in E$, where $l \not\in L_i$, an email $m_C \in E$ exists such that
  $m_C \prec f(i,l.m',G)_B$, $i \in R(m) \cup C$ and $l \in FI(m)$.
\end{enumerate}

Condition L.\ref{leg1} states that the agents can only forward messages 
they previously received.
Conditions L.\ref{leg2} and L.\ref{leg3} state that if an agent sends or forwards
a note that he did not initially hold, then he 
must have learnt it by means of an earlier email.  

So a state is legal if its emails can be partially ordered in such a
way that every forward is preceded by its original message, and for
every note sent in an email there is an explanation how the sender of
the email learnt this note.  As every partial ordering can be extended
to a linear ordering, the emails of a legal state can be ordered in
such a way that each agent has a linear ordering on its emails.
However, such a linear ordering does not need to be unique. For
example, the emails $s(i,l,j)_{\emptyset}$ and $s(i,l,k)_{\emptyset}$ can always
be ordered in both ways.

Moreover, a strict partial ordering that ensures that a state is legal
does not need to be unique either and incompatible minimal partial
orderings can exist.  Here is an example provided by one of the
referees.  Suppose that $l \in L_i \setminus L_j$ and $j \in G_1 \cap
G_2$, and consider the set of messages $\{s(i, l, G_1), \ s(i, l,
G_2), \ s(j, l, k)\}$. The resulting state (we identify here each
message $m$ with the email $m_{\emptyset}$) is legal. There are two minimal
spos that can be used to establish this, $s(i, l, G_1) \prec s(j, l,
k)$ and $s(i, l, G_2) \prec s(j, l, k)$.  So we cannot conclude that
any specific message sent by agent $i$ has to precede the message sent
by agent $j$, though we have to assume that at least one of them does.

This shows that the causal relation between emails essentially differs
from the causal relation between messages in distributed systems, as
studied in~\cite{Lam78}. Further, the assumption that communication is
synchronous does not result in a unique spo on the
considered emails. 

\section{Epistemic language and its semantics}
\label{sec:epistemic}

We want to reason about the knowledge of the agents after an email
exchange has taken place. For this purpose we use a language $\Lang$ of
communication and knowledge defined as follows:
\begin{eqnarray*}
\varphi &::=& m \mid i \blacktriangleleft m \mid \neg\varphi \mid \varphi \land \varphi \mid C_G \varphi
\end{eqnarray*}
Here $m$ denotes a message.
The formula $m$ expresses the fact that $m$ has been sent in the past,
with some unknown group of BCC recipients.
The formula $i \blacktriangleleft m$
expresses the fact that agent $i$ was involved in a full version of the message
$m$, i.e., he was either the sender, a recipient or a BCC recipient.
The formula $C_G \varphi$ denotes common knowledge of the formula
$\varphi$ in the group $G$.

We use the usual abbreviations $\lor$, $\rightarrow$ and $\leftrightarrow$ and use
$K_i \varphi$ as an abbreviation of $C_{\{i\}}\varphi$.
The fact that an email with a certain set of BCC recipients was sent
can be expressed in our language by the following abbreviation:
\[m_B ::= m \land \bigwedge_{i \in S(m) \cup R(m) \cup B} i \blacktriangleleft m \land
\bigwedge_{i \not\in S(m) \cup R(m) \cup B} \neg i \blacktriangleleft m\]
This formula expresses the fact that the message $m$ was
sent with exactly the group $B$ as BCC recipients, which captures
precisely the intended meaning of $m_B$.
The BCC recipients are distinguished from the regular recipients in $R(m)$ by the fact
that for any agent $i$ in $S(m)$ and $R(m)$, the fact that $i \blacktriangleleft m$ holds
follows from the fact that $m$ holds. On the other hand, for the agents in 
$B$, the fact that $i \blacktriangleleft m$ holds follows from $m_B$ and not from $m$ alone.

We now provide a semantics for this language interpreted on legal states, inspired by the 
epistemic logic and the history-based approaches of \cite{pacpar07}
and \cite{PR03}. For every agent $i$ we define an
indistinguishability relation $\sim_i$, where we intend $s \sim_i s'$
to mean that agent $i$ cannot distinguish between the states $s$
and $s'$. We first define this relation on the level of emails as
follows (recall that we assume that senders and regular recipients are not
BCC recipients):
\[m_B \sim_i m'_{B'}\]
iff one of the following contingencies holds:

\begin{enumerate}[(i)]
\item $i \in S(m)$, $m = m'$ and $B = B'$,

\item $i \in R(m) \setminus S(m)$ and $m = m'$,

\item $i \in B \cap B'$, and $m = m'$.
\end{enumerate}

Condition (i) states that the sender of an email confuses it only with the email itself.
In turn, condition (ii) states that each regular recipient of an email who is not
a sender confuses it with any email with the same message but possibly sent to
a different BCC group. Finally, condition (iii) states that each BCC recipient of an
email confuses it with any email with the same message but sent to a
possibly different BCC group of which he is also a member.  

\begin{example}
Consider the emails $e := s(i, l, j)_\emptyset$ and $e':= s(i,
l, j)_{\{k\}}$.  We have then $e \not\sim_i e'$, $e \sim_j e'$ and $e
\not\sim_k e'$. Intuitively, agent $j$ cannot distinguish between
these two emails because he cannot see whether $k$ is a BCC 
recipient. In contrast, agents $i$ and $k$ can distinguish between these two emails.
\end{example}

Next, we extend the indistinguishability relation to legal states by defining
\[(E, \L) \sim_i (E', \L')\]
iff all of the following hold:

\begin{itemize}
\item $L_i = L'_i$,
\item for every $m_B \in E$ such that $i \in S(m) \cup R(m) \cup B$ 
there is $m_{B'} \in E'$ such that $m_B \sim_i m_{B'}$,
\item for every $m_{B'} \in E'$ such that $i \in S(m) \cup R(m) \cup B$ 
there is $m_{B} \in E$ such that $m_B \sim_i m_{B'}$.
\end{itemize}

So two states cannot be distinguished by an agent if they agree on
his notes and their email sets look the same to him.  Since we assume
that the agents do not know anything about the other notes, we do not
refer to the sets of notes of the other agents.  Note that $\sim_i$ is
an equivalence relation.

\begin{example} \label{exa:3}
Consider the legal states
$s_1$ and $s_2$ which are identical apart from their sets of emails:
\[
\begin{array}{rcl}
E_{s_1} &:=& \{s(i, l, j)_\emptyset, f(j, s(i, l, j), o)_\emptyset\},\\
E_{s_2} &:=& \{s(i, l, j)_{\{k\}}, f(j, s(i, l, j), o)_\emptyset, f(k, s(i, l, j), o)_\emptyset\}. \\
\end{array}
\]

We assume here that $l \in L_i$ and that in each state
the emails are ordered by the textual ordering. 
So in the first state agent $i$ sends a message with note $l$ to agent $j$ and 
then $j$ forwards this message to agent $o$. Further, in the second
state agent $i$ sends the same message but with a BCC to agent $k$,
and then both agent $j$ and agent $k$ forward the message to agent
$o$.

From the above definition it follows that $s_1 \not\sim_i s_2$, $s_1
\sim_j s_2$, $s_1 \not\sim_k s_2$ and $s_1 \not\sim_o s_2$. For
example, the first claim holds because, as noticed above, $s(i, l,
j)_\emptyset \not\sim_i s(i, l, j)_{\{k\}}$.
Intuitively, in state $s_1$
agent $i$ is aware that he sent a BCC to nobody, while in state $s_2$
he is aware that he sent a BCC to agent $k$. In turn, in both states
$s_1$ and $s_2$ agent $j$ is aware that he received the message
$s(i, l, j)$ and that he forwarded the email $f(j,s(i, l, j),o)_\emptyset$.
Intuitively, in state $s_2$ agent $j$ does not notice the BCC of the message $s(i, l,
j)$ and is not aware of the email $f(k, s(i, l, j), o)_\emptyset$.
\end{example}

In order to express common knowledge, we define for a group of agents
$G$ the relation $\sim_G$ as the reflexive, transitive closure of
$\bigcup_{i \in G} \sim_i$.
Then we define the truth of a formula from our language in a state inductively as
follows, where $s = (E, \L)$:
\[\begin{array}{lll}
s \models m &\textrm{iff}&\exists B: m_B \in E\\
s \models i \blacktriangleleft m &\textrm{iff}&\exists B: m_B \in E  \textrm{ and } i \in S(m) \cup R(m) \cup B\\
s \models \neg\varphi  &\textrm{iff}&s \not\models \varphi\\
s \models \varphi \wedge \psi  &\textrm{iff}&s \models \varphi \textrm{ and } s \models \psi\\
s \models C_G \varphi  &\textrm{iff}&\textrm{$s' \models \varphi$ for every legal state $s'$ such that $s \sim_{G} s'$}\\
\end{array}\]

We say that $\varphi$ is \bfe{valid} (and often just write `$\varphi$'
instead of `$\varphi$ is valid') if for all legal states $s$, $s
\models \varphi$.  

Even though this definition does not specify the form of communication, one can deduce from
the definition of the relation $\sim$  that
the communication is synchronous, that is, that each email is simultaneously received by all the recipients.
We shall discuss this matter in more detail in Section~\ref{sec:dis}.
Note also that the condition of the form $m_B \in E$ present in the second clause
implies that for every email $m_B$ 
the following equivalence is valid for all $i,j \in S(m) \cup R(m) \cup B$:
\[
i \blacktriangleleft m \leftrightarrow j \blacktriangleleft m.
\]
This means that in every legal state $(E, \L)$ either all recipients of the email $m_B$
received it (when $m_B \in E$) or none (when $m_B \not\in E$).

The limited form of the introduced language implies
that the agents cannot simulate a common `clock' using which they
could deduce how many messages have been sent. Also, there is no
common `blackboard' using which they could deduce how many messages
have been sent by other agents between two consecutive messages they
have received. Further, the agents do not have a local `clock' using which
they could count how many messages they sent or received.

The following lemma clarifies when specific formulas are valid. In the
sequel we shall use these observations implicitly.
Below we use the relation \bfe{is part of} on messages, defined inductively as follows:

\begin{itemize}
\item  $m$ is part $f(i, l.m, G)$,

\item  if $m'$ is part of $m$, then $m'$ is part $f(i, l.m, G)$.
\end{itemize}

\begin{lemma} \label{lem:1}
 \mbox{}
 \begin{enumerate}[(i)]
 \item $m \rightarrow m'$ is valid iff $m = m'$ or $m'$ is part of the message $m$.

 \item $m \rightarrow i \blacktriangleleft m'$ is valid iff $i \in S(m') \cup R(m')$
or for some note $l$ and group $G$, $f(i, l.m', G)$ is part of the message $m$.
 \end{enumerate}
\end{lemma}

The second item states that $m \rightarrow i \blacktriangleleft m'$ is valid either
if $i$ is a sender or a receiver of $m'$ (in that case actually $i \blacktriangleleft m'$ is valid)
or $i$ forwarded the message $m'$. The latter is also possible if $i$ was a BCC receiver of $m'$. 
The claimed equivalence holds thanks to condition L.\ref{leg1}.

\begin{example}
To illustrate the definition of truth let us return to Example \ref{exa:3}.
In state $s_2$ agent $j$ does not know that agent $k$ received the
message $s(i, l, j)$ since he cannot distinguish $s_2$ from the state
$s_1$ in which agent $k$ did not receive this message.  So $s_2
\models \neg K_j k \blacktriangleleft s(i, l, j)$ holds. 

On the other hand, in every legal state $s_3$ such that $s_2 \sim_o
s_3$ both an email $f(k, s(i, l, j), o)_C$ and 
a `justifying' email $s(i, l, j)_B$ have to exist such that $s(i,
p, j)_B \prec f(k, s(i, l, j), o)_C$ and $k \in B$,
where $\prec$ is an spo such that 
the emails of $s_3$ satisfy conditions L.\ref{leg1}-L.\ref{leg3} w.r.t.
$\prec$.
Consequently
$s_3 \models k \blacktriangleleft s(i, l, j)$, so $s_2 \models K_o k
\blacktriangleleft s(i, l, j)$ holds, so by sending the forward agent
$k$ revealed himself to $o$ as a BCC recipient.

We leave to the reader checking that both $s_2 \models C_{\{k,o\}} 
k \blacktriangleleft s(i, l, j)$ and $s_2 \models
\neg C_{\{j,o\}} k \blacktriangleleft s(i, l, j)$ holds. In words,
agents $k$ and $o$ have common knowledge that agent $k$ was involved
in a full version of the message $s(i,l,j)$, while the agents $j$ and
$o$ don't.
\end{example}

\section{Epistemic contents of emails}
\label{sec:EI}

In Subsection \ref{subsec:legalstates} we defined the factual
information contained in a message.  Using the epistemic language
introduced in the previous section we can define the \bfe{epistemic 
information} contained in a message or an email.  First, we define
it for messages as follows:
\begin{eqnarray*}
EI(s(i, l, G)) &:=& C_{\{i\} \cup G} s(i, l, G),\\
EI(f(i, l.m, G)) &:=& C_{\{i\} \cup G} (f(i, l.m, G) \land EI(m)).
\end{eqnarray*}

So the epistemic information contained in a message is the statement
that the sender and receivers acquire common knowledge of the fact
that the message was sent.  In the case of a forward the epistemic
information contained in the original message also becomes common
knowledge.  This results in nested common knowledge. In general,
iterated forwards can lead to arbitary nestings of the common
knowledge operator, each time involving a different group of agents.

The definition of the epistemic information contained in an email
additionally needs to capture the information about the agents who are
on the BCC list of an email. We define: 
\[
EI(m_B) :=  EI(m) \land \bigwedge_{i \in B}C_{S(m) \cup \{i\}} (EI(m) \land i \blacktriangleleft m) \land C_{S(m)} m_B.
\]
So $EI(m_B)$ states that
\begin{itemize}
 \item the epistemic information contained in the message $m$ holds,
 \item the sender of the message and each separate agent on the BCC
   list have common knowledge of this epistemic information and of the
   fact that this agent received the message,
 \item the sender knows the precise set of BCC recipients.
\end{itemize}

The following result clarifies the nature of the epistemic information contained in a message or an email.

\begin{theorem} \label{thm:epis}
The following equivalences are valid:
\begin{enumerate}[(i)]
  \item $m \leftrightarrow EI(m)$,

  \item $m_B \leftrightarrow EI(m_B)$.
\end{enumerate}
\end{theorem}
\begin{proof}
Each relation $\sim_j$ on the level of states is an equivalence relation,
so for all formulas $\varphi$ and $G \subseteq \Ag$, the implication $C_G \varphi \rightarrow \varphi$,
and hence in particular $EI(m) \rightarrow m$ and $EI(m_B) \rightarrow m_B$, is valid.

$(i)$ To prove the validity of $m \rightarrow EI(m)$ take a
message $m$. Let $A = S(m) \cup R(m)$.  Consider an arbitrary legal
state $s$ and assume that $s \models m$.  Suppose $s \sim_A s'$ for
some legal state $s'$.  Then there is a path $s = s_0 \sim_{i_1} s_1
\sim_{i_2} \ldots \sim_{i_l} s_l = s'$ from $s$ to $s'$, where $i_1, \ldots,
i_l \in A$ and $l \geq 0$.

For every $k \in \{1, \ldots, l\}$ suppose $s_k = (E_k, \L_k)$.
Then for every $k \in \{1, \ldots, l\}$, $s_{k-1} \models m$ implies
that for some $B$, $m_B \in E_{k-1}$.  Now, since $i_k \in S(m) \cup
R(m)$, by the clauses $(i)$ and $(ii)$ of the definition of the
$\sim_{i_k}$ relation on the emails for some group $B'$ we have $m_{B'} \in E_k$,
which implies $s_k \models m$.  Since $s \models m$, an inductive
argument shows that $s' \models m$. This proves that $s \models C_A
m$.  So we established the validity of the implication
\[
m \rightarrow C_A m,
\]
and in particular of $s(i, l, G) \rightarrow EI(s(i, l, G))$.

For the forward messages we proceed by induction on the
structure of the messages.  Consider the message $f(i, l.m, G)$.  
The implication $f(i, l.m, G) \rightarrow m$ is valid, so by
the induction hypothesis the
implication $f(i, l.m, G) \rightarrow EI(m)$ is valid.
Since we showed already that the implication 
$f(i, l.m, G) \rightarrow C_{\{i\} \cup G} f(i, l.m, G)$ is valid, we conclude that the implication
$f(i, l.m, G) \rightarrow C_{\{i\} \cup G} (f(i, l.m, G) \wedge EI(m))$ is also valid.

$(ii)$
We already established the validity of $m \rightarrow EI(m)$. Then by the
definition of $m_B$ the implication $m_B \rightarrow EI(m)$ is
also valid.

Let $i \in B$. Consider an arbitrary legal state $s$ and assume that
$s \models m_B$.  Suppose $s \sim_{S(m) \cup \{i\}} s'$ for some legal
state $s'$.  Then there is a path $s = s_0 \sim_{i_1} s_1 \sim_{i_2}
\ldots \sim_{i_l} s_l = s'$ from $s$ to $s'$, where $i_1, \ldots, i_l \in
S(m) \cup \{i\}$ and $l \geq 0$.

For every $k \in \{1, \ldots, l\}$ suppose $s_k = (E_k, \L_k)$. 
Then for every $k \in \{1, \ldots, l\}$, 
$s_{k-1} \models m_B$ implies that $m_B \in E_{k-1}$ and then by the
definition of $\sim_k$, $m_{B'} \in E_k$ for some $B'$ such that $i \in
B'$. This means that $s_k \models i \blacktriangleleft m$ and $s_k \models
m$ which implies by $(i)$  $s_k \models EI(m)$. Since $s \models m_B$ an
inductive argument then shows that $s' \models EI(m) \land i
\blacktriangleleft m$.  So $s \models C_{S(m) \cup \{i\}} (EI(m) \land i
\blacktriangleleft m)$.  

Finally, suppose that $s \sim_{j} s'$, where $S(m) = \{j\}$, and $s \models m_B$.
By the definition of the $\sim_{j}$ relation
on the level of states $m_B \in E_{s'}$ so $s' \models m_B$.
This proves $s \models C_{S(m)} m_B$.

We conclude that the implication $m_B \rightarrow EI(m_B)$ is valid. Trivially, $EI(m_B) \rightarrow m_B$ is also valid.
\end{proof}

Using the above theorem we can determine
`who knows what' after an email exchange $E$ (taken from a legal state $(E, L)$)
took place.
The problem boils down to computing $\bigwedge_{e \in E} EI(e)$.
When we are interested in a specific fact, for example whether after an email exchange
$E$ took place agent $i$ knows a formula $\psi$, we simply need to
establish the validity of the implication $\bigwedge_{e\in E} EI(e)
\rightarrow C_i \psi$.

Using the epistemic information contained in an email we can define the \bfe{information gain}
of an agent resulting from sending or receiving of an email as follows. Suppose $i \in S(m) \cup R(m) \cup B$. Then

\[
IG(m_B, i) := \left\{
\begin{array}{ll}
EI(m_B) & \textrm{if } S(m) = \{i\} \\
EI(m)   & \textrm{if }  i \in R(m) \\
C_{S(m) \cup \{i\}} (EI(m) \land i \blacktriangleleft m) & \textrm{if } i \in B
\end{array} \right.
\]

We have then the following immediate consequence
of Theorem \ref{thm:epis}.

\begin{corollary}
Take a legal state $s = (E,\L)$ and an email $m_B \in E$.
Then for every agent $i \in S(m) \cup R(m) \cup B$,
\[
s \models K_i IG(m_B, i).
\]
\end{corollary}

\begin{example}
Using the notion of an information gain we can 
answer the first question posed in Example \ref{exa:1},
namely what Alma learned from Bob's email.
First recall the notation introduced at the end of  
Subsection \ref{subsec:emails}:

$m := s(c, l, \{a,d\})$,

$e_1 := m'_{\emptyset}$, where $m' := f(a, m, b)$,
and

$e_2 := m''_{\{a\}}$, $m'' := f(b, m', \{c, d\})$.

By definition
\begin{eqnarray*}
EI(m) &=& C_{\{a,c,d\}} m,\\
EI(m') &=& C_{\{a,b\}} (m' \land EI(m)),\\
EI(m'') &=& C_{\{b,c,d\}} (m'' \land EI(m')),\\
IG(e_2, a) &=& C_{\{a,b\}} (EI(m'') \land b \blacktriangleleft m'').
\end{eqnarray*}
This should be contrasted with the information Alma had after she sent the email $e_1$,
which was $EI(m')$.
\end{example}

\section{Common knowledge}
\label{sec:common}

Our main objective is to clarify when a group of agents acquires common
knowledge of the formula expressing that
an email was sent. This can be done within
our framework, which shows that it is appropriate for
investigating epistemic consequences of email exchanges.

First, given a set of emails $E$ and a group of agents $A$, we define
\[E_A := \{m_B \in E \mid A \subseteq S(m) \cup R(m) \ or \ \exists j \in B: (A \subseteq S(m) \cup \{j\})\}.\]

When $e \in E_A$ we shall say that the email $e$ is \bfe{shared by the group $A$}. 
Note that when $|A| \geq 3$, then $e \in E_A$ iff $A \subseteq S(m) \cup R(m)$.
When $|A| = 2$, then $e \in E_A$ also when $\exists j \in B: A = S(m) \cup \{j\}$,
and when $|A| = 1$, then $e \in E_A$ also when $A = S(m)$ or $\exists j \in B: A = \{j\}$.

The following theorem summarizes our results. It provides a simple way of testing
whether a message or an email is a common knowledge of a group of agents.

\begin{theorem}[Main Theorem]
Consider a legal state $s = (E, \L)$ and a group of agents $A$.

\begin{enumerate}[$(i)$]
\item $s \models C_A m$ iff  there is $m'_{B'} \in E_A$ such that $m' \rightarrow m$ is valid.
\item Suppose that $|A| \geq 3$. Then $s \models C_A m_B$ iff the following hold, where, recall,
$\Ag$ is the set of agents:
\begin{description}
 \item[C1] $\Ag = S(m) \cup R(m) \cup B$,

 \item[C2] for each $i \in B$ there is $m'_{B'} \in E_A$ such that $m' \rightarrow i \blacktriangleleft m$ is valid,

 \item[C3] there is $m'_{B'} \in E_A$ such that $m' \rightarrow m$ is valid.
\end{description}
\end{enumerate}
\end{theorem}

Part $(i)$ show that when we limit our attention to messages, then things are as expected:
a group of agents acquires common knowledge of a message $m$ iff they receive an email
a part of which is $m$. If we limit our presentation to emails with the empty 
BCC sets we get as a direct corollary the counterpart of this result for a simplified
framework with messages only.

To understand part $(ii)$ note that it states that
$s \models C_A m_B$ iff
\begin{itemize}
\item the email $m_B$ involves all agents,

\item for every agent $i$ that is on the BCC list of $m_B$ there is an
email shared by the group $A$ that proves that $i$ forwarded message $m$,

\item there is an email shared by the group $A$ that proves the existence 
  of the message $m$.
\end{itemize}
The first of the above three items is striking and shows that common
knowledge of an email is rare.  \textbf{C3} is just the condition used
in part $(i)$. So an email $m_B$ such that $A \subseteq S(m) \cup R(m)$
does ensure that the group of agents $A$ acquires common knowledge of
$m$.  However, the group $A$ can never know what was the set of the BCC
recipients of $m_B$ \emph{unless} it was the set $\Ag \setminus (S(m)
\cup R(m))$ \emph{and} there is a proof for this fact in the form of
the `disclosing emails' from all members of $B$.

Having in mind that the usual purpose of the BCC is just to inform its
recipients of a certain message (that they are supposed to `keep for
themselves'), we can conclude that the presence of the BCC feature
essentially precludes the possibility that a group of agents can acquire common
knowledge of an email. Informally, the fact that the BCC feature
creates `secret information' has as a consequence that common
knowledge of an email is only possible if this secret information is
completely disclosed to the group in question.  Moreover, the message
has to be sent to all agents.

Note that using the notion of the information gain introduced in the previous section
we can determine for each agent in a group $A$ what he learned from a message $m$ or
an email $m_B$. In some circumstances, like when $m = s(i,l,G)$ and $A \subseteq G \cup \{i\}$,
this information gain can imply $C_A m$. However, the definition of $EI(m_B)$ implies
that the information gain can imply $C_A m_B$ only in the obvious case when $A = S(m)$.

Finally, the above result crucially depends on the fact that the notes
are uninterpreted.  If we allowed emails that contain propositional
formulas of the language $\Lang$ from Section \ref{sec:epistemic}
augmented by the notes, then an agent could communicate to a group $A$
the fact that he sent an email $m_B$ (with a precise set of the BCC
recipients). Then $m_B$ would become a common knowledge of the group
$A$.

As an aside let us mention that there is a corresponding result for
the case when $|A| < 3$, as well.  However, it involves a tedious case
analysis concerning the possible relations between $A, S(m), R(m)$ and
$B$, so we do not present it here.  

\begin{example}
We can use the above result to answer the second question posed in Example \ref{exa:1}.
Let $s$ be the state whose emails consist of the considered four emails, so

$e_0 := m_{\emptyset}$, where $m := s(c, l, \{a,d\})$,

$e_1 := m'_{\emptyset}$, where $m' := f(a, m, b)$,

$e_2 := m''_{\{a\}}$, where $m'' := f(b, m', \{c, d\})$, and

$e_3 := f(a, m'', \{c, d\})_{\{b\}}$.

Alma's set of notes in $s$ consists of $l$ while the sets of notes of Bob, Clare and Daniel are empty.
Note that $s$ is legal.
We have then
\[
s \not\models C_{\{a,b,c,d\}} s(c, l, \{a,d\}).
\]
The reason is that we have 
\[E_{\{a,b,c,d\}} = \emptyset.\]
Indeed, for no $m^* \in \{m,m',m'',f(a,m'',\{c,d\})\}$ we have \[\{a,b,c,d\} \subseteq S(m^*) \cup R(m^*)\]
and for no $m^*_B \in \{e_0, e_1, e_2, e_3\}$  we have some $j \in B$ such that \[\{a,b,c,d\} \subseteq S(m^*) \cup \{j\}.\] So there are 
no messages that ensure common knowledge in the group $\{a,b,c,d\}$. So even though there have been three forwards of the original message, it is not common knowledge.

Clearly, if the original message $s(c, l, \{a,d\})$ is not common knowledge then its forward $f(a,m,b)$ is not common knowledge either. Another way to derive this is directly from the Main Theorem. Namely, we have
\[
s \not\models C_{\{a,b,c,d\}} f(b, m', \{c,d\})_{\{a\}}.
\]
The reason is that condition \textbf{C2} does not hold since no email shared by 
$\{a,b,c,d\}$ exists that proves that Alma received $m''$. In contrast, 
\[
s \models C_{\{a,c,d\}} f(b, m', \{c,d\})_{\{a\}}
\]
does hold, since the email $e_3$ is shared by $\{a,c,d\}$. 
Further, if Alma had used the forward $f(a, m'', \{b,c,d\})_{\emptyset}$, then condition \textbf{C2}
would hold and we could conclude for this modified state $s'$ that
\[
s' \models C_{\{a,b,c,d\}} f(b, m', \{c,d\})_{\{a\}}.
\]
\end{example}

\section{Proof of the Main Theorem}
\label{sec:proof}
We establish first a number of auxiliary lemmas.
We shall use a new strict partial ordering on emails.
We define
\[m_B < m'_{B'}\textrm{ iff }m \neq m'\textrm{ and }m' \rightarrow m.
\]

Note that by Lemma \ref{lem:1} $m \neq m'$ and
$m' \rightarrow m$ precisely if $m'$ is a forward, or a forward of
a forward, etc, of $m$. Then for two emails $m_B$ and $m'_{B'}$ from a
legal state $s$ that satisfies conditions L.\ref{leg1}-L.\ref{leg3}
w.r.t.~an spo $\prec$, $m_B < m'_{B'}$ implies $m_B \prec m'_{B'}$ on
the account of condition L.\ref{leg1}.  However, the converse does not
need to hold since $m_B \prec m'_{B'}$ can hold on the account of
L.\ref{leg2} or L.\ref{leg3}.  Further, note that the $<$-maximal
elements of $E$ are precisely the emails in $E$ that are not
forwarded.

Given a set of emails $E$ and $E' \subseteq E$ we then define the \bfe{downward
closure} of $E'$ by 
\[
E'_{\leq} := E' \cup \{e \in E \mid \exists e' \in E': e < e'\}.
\]
The set of emails $E$ on which the downward closure of $E'$ depends will always be clear from the context.

Next, we introduce two operations on states.  Assume a state $(E,\L)$ and an email $m_B \in E$.

We define the state
\[s\setminus m_B := (E \setminus \{m_B\}, \L'),\]
with
\[L'_i := \left\{ 
\begin{array}{ll}
L_i \cup FI(m)  &\textrm{if }i \in R(m) \cup B \\
L_i &  \textrm{otherwise}
\end{array}
\right.\]

Intuitively, $s\setminus m_B$ is the result of removing the email
$m_B$ from the state $s$, followed by augmenting the sets of notes of
its recipients in such a way that they initially already had the notes
they would have acquired from $m_B$. Note that $s\setminus m_B$ is a
legal state if $m_B$ is an $<$-maximal element of $E$.

Next, given $C \subseteq B$ we define the state
\[s[m_{B \mapsto C}] := (E \setminus \{m_B\} \cup \{m_C\}, \L'),\]
with 
\[
L'_i := \left\{ 
\begin{array}{ll}
L_i \cup FI(m)  &  \textrm{if }i \in B \setminus C \\
L_i &  \textrm{otherwise}
\end{array}
\right. 
\]

Intuitively, $s[m_{B \mapsto C}]$ is the result of shrinking the
set of BCC recipients of $m$ from $B$ to $C$, followed by an
appropriate augmenting of the sets of notes of the agents that 
no longer receive $m$.

Note that $s[m_{B \mapsto C}]$ is a legal state if there is no
forward of $m$ by an agent $i \in B \setminus C$, i.e., no email of
the form $f(i,l.m,G)_D$ exists in $E$ such that $i \in B \setminus C$.

We shall need the following lemma that clarifies the importance of 
the set $E_A$ of emails.

\begin{lemma} \label{lem:sim}
 Consider a legal state $s = (E, \L)$ and a group of agents $A$.
Then for some $\L'$ the state $s' := ((E_A)_{\leq}, \L')$ is legal
and $s \sim_A s'$.
\end{lemma}
\begin{proof}
We prove that for all $<$-maximal emails $m_B \in E$ such that
$m_B \not\in E_A$ (so neither $A \subseteq S(m) \cup R(m)$ nor
$\exists i \in B: (A \subseteq S(m) \cup \{i\})$)
we have $s \sim_A s\setminus m_B$.
Iterating this process we get the desired conclusion.

Suppose $m_B$ is a $<$-maximal email in $E$ such that $m_B \not\in
E_A$.  Take some $j \in A \setminus (S(m) \cup R(m))$.
Suppose first $j \not\in B$. Then $s \sim_j s \setminus m_B$ so $s
\sim_A s \setminus m_B$.

Suppose now $j \in B$. 
Define 
\[
s_1 := s[m_{B \mapsto \{j\}}].
\]
Then $s_1$ is a legal state and $s \sim_j s_1$. Next, define
\[
s_2 := s[m_{B \mapsto \emptyset}].
\]
Now take some $k \in A \setminus (S(m) \cup \{j\})$. Then $s_1 \sim_k s_2 \sim_j s \setminus m_B$ so $s \sim_A s \setminus m_B$.
Note that both $s_1$ and $s_2$ are legal states since $m_B$ is $<$-maximal.
\end{proof}

Using the above lemma we now establish two auxiliary results concerning common knowledge of
the formula $i \blacktriangleleft m$ or of its negation.

\begin{lemma} \label{lem:im}
\mbox{} 
\begin{enumerate}[(i)]
\item 

\begin{tabbing}
$s \models C_A i \blacktriangleleft m$ iff \= $\exists m'_B \in E_A : (m' \rightarrow i \blacktriangleleft m)$ \\
                                           \> or ($A \subseteq S(m) \cup \{i\}$ and $\exists m_B \in E_A : (i \in B)$).
\end{tabbing}

\item 
$s \models C_A \neg i \blacktriangleleft m\textrm{ iff }s \models \neg \: i \blacktriangleleft m\textrm{ and }(A \subseteq S(m) \cup \{i\}\textrm{ or }s \models C_A \neg m)$.

\end{enumerate}
\end{lemma}

To illustrate various alternatives listed in $(i)$ note that
each of the following emails in $E$ ensures that $s \models C_{\{j\}}
i \blacktriangleleft m$, where in each case $m$ is the corresponding
send message:
\[
\begin{array}{c}
s(i,l,G)_{\{j\}}, f(k, q.s(i,l,G), H)_{\{j\}},\ \\
s(k,l,i)_{\{j\}}, \ f(i, q.s(k,l,G), H)_{\{j\}}, \ s(j,l,G)_{\{i\}}.
\end{array}
\]
The first four of these emails imply $s \models C_{\{j\}} i \blacktriangleleft m$ by the first clause of $(i)$,
the last one by the second clause.

\begin{proof}
$(i)$
$(\Rightarrow)$
Suppose $s \models C_A i\blacktriangleleft m$. 
Take the legal state $s'$ constructed in Lemma
\ref{lem:sim}. Then $s \sim_A s'$, so 
$s' \models i\blacktriangleleft m$.  

Hence for some group $B$ we have $m_B \in (E_A)_{\leq}$ and $i \in S(m) \cup R(m)
\cup B$. Three cases arise.

\emph{Case 1}. $i \in S(m) \cup R(m)$.

Then $m \rightarrow i\blacktriangleleft m$.  So if $m_B \in E_A$, then the claim
holds.  Otherwise some email $m'_{B'} \in E_A$ exists such that $m_B <
m'_{B'}$.  Consequently $m' \rightarrow m$ and hence $m' \rightarrow
i\blacktriangleleft m$. So the claim holds as well.

\emph{Case 2}. $i \not\in S(m) \cup R(m)$ and $A \subseteq S(m) \cup \{i\}$.

Then $i \in B$ since $i \in S(m) \cup R(m) \cup B$.
Then by the definition of $E_A$, $m_B \in E_A$ so the claim holds.

\emph{Case 3}. $i \not\in S(m) \cup R(m)$ and $\neg (A \subseteq S(m) \cup \{i\})$.

If for some note $l$ and groups $G$ and $C$ we have $f(i,l.m,G)_C \in
(E_A)_{\leq}$, then either $f(i,l.m,G)_C \in E_A$ or for some $m'_{B'}
\in E_A$ we have $f(i,l.m,G)_C < m'_{B'}$. In the former case we use
the fact that the implication $f(i,l.m,G) \rightarrow i\blacktriangleleft m$
is valid. In the latter case $m' \rightarrow f(i,l.m,G)$ and hence $m'\rightarrow
i\blacktriangleleft m$.  So in both cases the claim holds.

Otherwise let $s'' = s'[m_{B \mapsto B \setminus \{i\}}]$.
Note that $s''$ is a legal state because $i$ does not forward $m$ in $s'$. 
Take some $j \in A \setminus (S(m)  \cup \{i\})$. Then 
$s' \sim_j s''$, so $s \sim_A s''$. Moreover, $s'' \models \neg i\blacktriangleleft m$, 
which yields a contradiction. So this case cannot arise.

$(\Leftarrow)$ The claim follows directly by the definition of semantics.  We
provide a proof for one representative case.  Suppose that for some
email $m'_B \in E_A$ both $A \subseteq S(m') \cup R(m')$ and $m' \rightarrow i
\blacktriangleleft m$.  Take some legal state $s'$ such that $s \sim_A
s'$.  Then for some group $B'$ we have $m'_{B'} \in E_{s'}$. So $s'
\models m'$ and hence $s' \models i \blacktriangleleft m$.
Consequently $s \models C_A i\blacktriangleleft m$. 

$(ii)$ Let $s = (E,\L)$. 

$(\Rightarrow)$ Suppose $s \models C_A \neg i \blacktriangleleft m$. Then $s \models \neg i \blacktriangleleft m$. 
Assume $A \not\subseteq S(m) \cup \{i\}$ and $s \not\models C_A \neg m$. Then there is some legal state $s' = (E',\L')$
such that $s \sim_A s'$ and $s' \models m$. Then there is some group $B$ such that $m_B \in E'$.
Let $j \in A \setminus (S(m) \cup \{i\})$ and let $s'' = (E' \setminus \{m_B \} \cup \{m_{B \cup \{i\}}\}, \L')$.
Then $s' \sim_j s''$ so $s \sim_A s''$. But $s'' \models i \blacktriangleleft m$ which contradicts our assumption.

$(\Leftarrow)$ 
Suppose that $s \models \neg i \blacktriangleleft m$ and either $A \subseteq S(m) \cup \{i\}$ or $s \models C_A \neg m$.
We first consider the case that $A \subseteq S(m) \cup \{i\}$. Let $s'$ be any legal state such that $s \sim_A s'$.
Assume $s' \models i \blacktriangleleft m$. Then $m_B \in E_{s'}$ for some group $B$ such that $i \in B$.
Since $A \subseteq S(m) \cup \{i\}$, any legal state $s''$ such that $s' \sim_A s''$ contains an email $m_C \in E_{s''}$
for some group $C$ such that $i \in C$. So $s'' \models i \blacktriangleleft m$.
In particular, this holds for the state $s$, which contradicts our assumption.
So $s' \models \neg s(i,n,G)$ and hence $s \models C_A \neg s(i,n,G)$.

Now we consider the case that $s \models C_A \neg m$. Let $s'$ be such that $s \sim_A s'$. Then $s' \models \neg m$.
Since $i \blacktriangleleft m \rightarrow m$ is valid, we get
$s' \models \neg i \blacktriangleleft m$. So $s \models C_A \neg i \blacktriangleleft m$.
\end{proof}

We are now ready to prove the Main Theorem.

\begin{proof}[of the Main Theorem]

$(i)$
$(\Rightarrow)$
Suppose $s \models C_A m$.  Take the legal state $s'$ constructed in
Lemma \ref{lem:sim}. Then $s \sim_A s'$, so $s' \models m$.  So for
some group $B$ we have $m_B \in (E_A)_{\leq}$.  
Hence either $m_B \in E_A$ or some email $m'_{B'} \in E_A$
exists such that $m_B < m'_{B'}$.
In both cases the claim holds.

$(\Leftarrow)$
Suppose that for some email $m'_B \in E_A$ we have $m' \rightarrow m$.  Take
some legal state $s'$ such that $s \sim_A s'$. Then by the form of $E_A$
and the definition of semantics for some group $B'$ we have $m'_{B'} \in
E_{s'}$. So $s' \models m'$ and hence $s' \models m$. Consequently $s 
\models C_A  m$.

$(ii)$
By the definition of $m_B$, the
fact that the $C_A$ operator distributes over the conjunction, part $(i)$  of the Main Theorem and
Lemma \ref{lem:im} we have
\[s \models C_A m_B\textrm{ iff  \textbf{C3}-\textbf{C6},}\]
where

\begin{description}
\item[C4] $\bigwedge_{i \in S(m) \cup R(m) \cup B}$ $((A \subseteq S(m) \cup \{i\}$ and $\exists B' : (m_{B'} \in E_A \textrm{ and } i \in B'))$ 
                                   or $\exists m'_{B'} \in E_A : (m' \rightarrow i \blacktriangleleft m))$,

\item[C5] $\bigwedge_{i \not\in S(m) \cup R(m) \cup B}$ $(A \subseteq S(m) \cup \{i\}$ or $s \models C_A \neg m)$,

\item[C6] $s\models \bigwedge_{i \not\in S(m) \cup R(m) \cup B} \neg i \blacktriangleleft m$.
\end{description}

$(\Rightarrow)$
Suppose $s \models C_A m_B$. Then properties \textbf{C3}-\textbf{C6} hold.
But $|A| \geq 3$ and $s \models C_A m$ imply that no conjunct of \textbf{C5} holds. 
Hence property \textbf{C1} holds.

Further, since $|A| \geq 3$ the first
disjunct of each conjunct in \textbf{C4} does not hold.  So the second
disjunct of each conjunct in \textbf{C4} holds, which implies property \textbf{C2}.

$(\Leftarrow)$
Suppose properties \textbf{C1}-\textbf{C3}
hold. It suffices to establish properties \textbf{C4}-\textbf{C6}.

For $i\in S(m) \cup R(m)$ we have $m \rightarrow i \blacktriangleleft
m$. So \textbf{C2} implies property \textbf{C4}.
Further, since \textbf{C1} holds,
properties \textbf{C5} and \textbf{C6} hold vacuously. 
\end{proof}

\section{Analysis of BCC}
\label{sec:bcc}

In our framework we built emails out of messages using the BCC feature.
So it is natural to analyze whether and in what sense the emails can be reduced
to messages without BCC recipients.

Given a send email $s(i,l,G)_B$, where $B = \{j_1, ..., j_k\}$,
we can simulate it by the following sequence of
messages:
\[
s(i,l,G), f(i,s(i,l,G),j_1), ..., f(i,s(i,l,G),j_k).
\]
Analogous simulation can be formed for the forward email $f(i,l.m,G)_B$.
At first sight, it seems that this simulation has exactly the same epistemic 
effect as the original email with the BCC recipients.
In both states, each agent $j_1, ..., j_k$ receives separately
a copy of the message and only the sender of this message is aware of this.
However, there are two subtle differences.

First of all, there is a syntactic difference between message that
agents $j_1, ..., j_k$ receive in the original case and in the
simulation. In the original case they receive exactly the message $m$,
and in the simulation they receive a forward of it.  This also means
that if they reply to or forward the message, there is a syntactic
difference in this reply or forward. This difference is purely
syntactic and does not essentially influence the knowledge of the
agents, even though it clearly influences the truth value of the
formula $j \blacktriangleleft m$ which is true for $j \in \{j_1, ...,
j_k\}$ in the original case but not in the simulation.

The second difference is more fundamental.  If agents $j_1, ..., j_k$
are BCC recipients of $m$ and they do not send a reply to or a forward
of $m$, then each of them can be sure that no other agent but the
sender of $m$ knows he was a BCC recipient.  Indeed, in our framework
there is no message the sender of $m$ could send to another agent,
that expresses that agents $j_1, ..., j_k$ were the BCC recipients of
$m$.

In the case of the simulation however, these recipients do not receive a BCC
but a forward. Since these forwards may have additional BCC recipients of which
agents $j_1, ..., j_k$ are unaware, they cannot be sure that the other agents
do not know that they received a forward of the message.
Furthermore, the sender of $m$ could also forward the forward he sent to 
$j_1, ..., j_k$ without informing them about it, thus also revealing their
knowledge of $m$.

A concrete example that shows this difference is the following. 
\begin{example}
\label{exa:5}
Let
\[
E_s = \{s(1, l, 2)_{\{3\}}\}.
\]
Then $s \models K_3 \neg K_2 K_3 s(1,l,2)$, that is, agent 3 is sure
that agent 2 does not know about his knowledge of the message $s(1,l,2)$.
A simulation of this email without a BCC recipient would result in the state $t$ with
(we abbreviate here each email $m_{\emptyset}$ to $m$)
\[
E_{t} = \{s(1, l, 2), f(1, s(1,l,2), 3)\}.
\]
Now consider a state $t'$ with:
\[
E_{t'} = \{s(1,l,2), f(1,s(1,l,2),3),f(1,f(1,s(1,l,2),3),2)\}.
\]
Clearly $t \sim_3 t'$ and $t' \models K_2K_3s(1,l,2)$. 
This shows that $t \not\models K_3 \neg K_2K_3s(1,l,2)$.
\end{example}

This argument can be made more general as follows.
Below, in the context of a state we identify each message $m$ with the email $m_{\emptyset}$.
Then we have the following result.

\begin{theorem}
Take a legal state $s = (E,\L)$, an email $m_B \in E$ and an agent $j \in B$ 
such that $E$ does not contain a forward of $m$ by $j$ or to $j$. 
Then for every set of messages $M$ such that $(M,\L)$ is a legal state we 
have for every agent $k \not\in S(m) \cup \{j\}$
\[s \models K_j m \land K_j \neg K_k K_j m,\]
while
\[(M,\L) \not\models K_j m \land K_j \neg K_k K_j m.\]
\end{theorem}
\begin{proof}
Agent $j$ is a BCC recipient of $m$ in $s$, so by the definition of the 
semantics $s \models K_j m$. We will first show that 
$s \models K_j \neg K_k K_j m$.
Take some state $t$ such that $s \sim_j t$. Then by the definition of the 
semantics there is some group $C$ such that $m_C \in E_t$ and $j \in C$.
Suppose that $m$ is a send email, say $m = s(i, l, G)$.
For the case that $m$ is a forward email the reasoning is analogous.
Let $u$ be the state like $t$, but with 
\[E_u = E_t \backslash \{s(i,l,G)_C\} \cup \{s(i,l,G)_{C \backslash \{j\}}, s(i, l, j)\}.\]
Note that we implicitly assume that no full version of $s(i, l, j)$
is already present in $E_t$. If there were such a full version, we could
do the same construction without adding $s(i,l,j)$ to $E_t$.

Since there are no forwards of $m$ by $j$ or to $j$ in $E$, and $s \sim_j t$,
there are no forwards of $m$ by $j$ or to $j$ in $E_t$. This shows that
$u$ is a legal state and that there are no forwards of $m$ to $j$ in
$E_u$ so $u \not\models K_j m$.
Clearly, for every $k \not\in S(m) \cup  \{j\}$ we have $t \sim_k u$. 
So $t \not\models K_k K_j m$, which shows that
 $s \models K_j \neg K_k K_j m$.

Take now any set of messages $M$ such that $(M,\L)$ is legal and suppose 
$(M,\L) \models K_j m$. Then by the Main Theorem there is some
message $m'$ in which agent $j$ was involved that implies that message $m$
was sent. By the requirements on the legal states we know that there is
such a message $m'$ of which agent $j$ was a recipient, and not
the sender, since agents can only send information they initially knew
or received through some earlier message. Since there are no BCC recipients
in $M$, we conclude that agent $j$ is a regular recipient of $m'$
that he received from some other agent and that $m' \rightarrow m$ is valid.

Define the set of messages $M'$ by
\[M' := M \cup \{f(S(m'), m', k)\}.\]
Note that $(M',\L)$ is a legal state, and $(M',\L) \models K_k m'$.
Since $j$ is a regular recipient of $m'$, $m' \rightarrow K_j m'$ is valid and 
since $m' \rightarrow m$ is also valid this implies that $(M',\L) \models K_k K_j m$.
Also, since $j$ is not involved in $f(S(m'),m',k)$, $(M,\L) \sim_j (M',\L)$.
This shows that $(M,\L) \not\models K_j \neg K_k K_j m$.
In view of our assumption that $(M,\L) \models K_j m$ we conclude
that $(M,\L) \not\models K_j m \land K_j \neg K_k K_j m$.
\end{proof}

In this theorem we assume that for the BCC recipient $j$ of the message $m$ 
there are no forwards of $m$ to $j$ or by $j$. The theorem shows that under
these assumptions, $s$ and $(M,L)$ can be distinguished by an epistemic formula
concerning the message $m$. We will now show that these assumptions are necessary.

\begin{example}
Take a legal state
$s = (E,\L)$ with \[E = \{s(1,l,2)_{\{3\}}, f(2,s(1,l,2),3)\}\]
and \[M = \{s(1,l,2), f(1,s(1,l,2),3), f(2,s(1,l,2),3)\}.\]
We can see that $(M,\L)$ is a perfect BCC-free simulation of $s$: 
for every formula $\varphi$ that holds in $s$, if we replace the occurrences
of $3 \blacktriangleleft s(1,l,2)$ in $\varphi$ by $f(1, s(1,l,2), 3)$ then the 
result holds in $(M,\L)$.
The reason that we can find such a set $M$ is that in $E$ there is a forward 
of $s(1,l,2)$ to agent $3$. This reveals the `secret' that 
agent 3 knows about $s(1,l,2)$ and then the fact that agent 3 was a BCC recipient
of $s(1,l,2)$ is no longer relevant.
\end{example}

\begin{example}
A similar example shows the importance of the assumption
that there are no forwards by a BCC recipient. 
Take a legal state $s = (E,\L)$ with 
\[E = \{s(1,l,2)_{\{3\}}, f(3,s(1,l,2),2)\}\]
and \[M = \{s(1,l,2), f(1,s(1,l,2),3), f(3,f(1,s(1,l,2),3),2)\}.\]
Again, for every formula $\varphi$ that holds in $s$, if we replace the occurrences
of $3 \blacktriangleleft s(1,l,2)$ in $\varphi$ by $f(1, s(1,l,2), 3)$ then the 
result holds in $(M,\L)$. Now the reason is that agent 3 informed agent 2
that he was a BCC recipient of $s(1,l,2)$ in $s$ by sending a forward of this message,
so again the fact that agent 3 knows $s(1,l,2)$ is not a secret anymore.
\end{example}

It is interesting to note that the impossibility of simulating BCC by means of messages
is in fact caused by our choice of uninterpreted notes as the basic content of the messages.
If our framework allowed one to send messages containing more complex information,
for example a formula of the form $j \blacktriangleleft m$, the sender of $m$ could have 
informed other agents who were the BCC recipients.
Then in Example \ref{exa:5} we could consider a state $s'$ with
\[
E_{s'} = \{s(1, n, 2)_{\{3\}}, s(1, 3 \blacktriangleleft s(1,n,2),2)\}.
\]
By appropriately extending our semantics we would have then $s \sim_3
s'$ and $s' \models K_2 K_3 s(1,n,2)$, and hence $s \not\models K_3 \neg
K_2 K_3 s(1,n,2)$, so the difference between the above two states $s$ and $t$ 
would then disappear.

Similarly, if we allowed epistemic formulas as contents of the
messages, then in the above example agent 1 could use the message
$s(1, K_3 s(1,n,2),2)$ to inform agent 2 that he BCC'ed agent 3 when
sending the message $s(1,n,2)$.  We leave an analysis of such
extensions of our framework and the role of BCC in these extended
settings as future work.

Finally, let us mention another feature of our syntax that cannot be
faithfully simulated by simpler means ---that of appending a note to a
forwarded message.  Suppose that we allow instead only a `simple'
forward $f(i,m,G)$ and simulate the current forward $f(i,l.m,G)$ by a
send and a simple forward, i.e., by the sequence $s(i,l,G), f(i,m,G)$.
Then the fact that the note $l$ was `coupled' with $m$ can in some
circumstances provide a piece of additional information that becomes
lost during the simulation.  Here is a concrete example. We do not use
BCC here, so each email $m_{\{\emptyset\}}$ is written as $m$.
\begin{example}
\label{exa:6}
Suppose that $l_1, l_2 \in L_1$ and $l_1, l_2 \not\in L_i$ for $i \neq 1$. Let $m := s(1, l, 1)$ and
\[
E_s := \{m, f(1, l_{2}.m, 2)\}.
\]
Then for all $i$ we have
$s \models K_1 (K_i m \rightarrow K_i l_2)$, that is, agent 1 knows that
every agent who knows the message $m$ also knows the note $l_2$.

A simulation of these two messages with a simple forward 
would yield the state $t$ with
\[
E_t := \{m, s(1, l_{2}, 2), f(1, m, 2)\}.
\]
Now consider a state $t'$ with:
\[
E_{t'} := \{m, s(1, l_{2}, 2), f(1, m, 2), f(2, m, 3)\}.
\]
Clearly $t \sim_1 t'$ and $t' \models K_3 m \wedge \neg K_3 l_2$.
This shows that $t \not\models K_1 (K_3 m \rightarrow K_3 l_2)$.
\end{example}

Note that this example exploits the fact that in our framework
the agents can forward the notes that are `buried'
within the received emails (thanks to the references to $l \in FI(m)$
in conditions L.\ref{leg2} or L.\ref{leg3} in Subsection \ref{subsec:legalstates}),
whereas they can only forward the messages they received. That is, they cannot
forward messages that are `buried' within the emais they received. This
natural restriction is satisfied by the email systems.

\section{Email exchanges}
\label{sec:dis}

Finally, we return to the issue of the synchronicity of the email
communication mentioned in Section~\ref{sec:epistemic}. To this end we
introduce an operational semantics that also allows us to provide a
characterization of the notion of a legal state in terms of email
exchanges. In this setting emails are sent in a nondeterministic
order, each time respecting the restrictions imposed by the legality
conditions L.\ref{leg1} -- L.\ref{leg3} of Subsection
\ref{subsec:legalstates}.

This operational semantics is defined in the style of \cite{Plo82}, 
though with some important differences concerning the notions of a
program state and the atomic transitions.
Let $M$ be the set of all messages (so \emph{not} emails).  By
a \bfe{mailbox} we mean a function $\sigma: \Ag \rightarrow {\cal P}(M)$;
$\sigma(i)$ is then the mailbox of agent $i$.  
If for all $i$ we have $\sigma_0(i) = \emptyset$, then we call $\sigma_0$ the
\bfe{empty mailbox}.  A \bfe{configuration} is a construct of the form
$<s,\sigma>$, where $s$ is a legal state and $\sigma$ is a mailbox.

Atomic transitions between configurations are of the form
\[
<s, \sigma> \rightarrow <s',\sigma'>,
\]
where $\dot{\cup}$ denotes disjoint union and

\begin{itemize}
\item $s := (E ~\dot{\cup}~ \{m_B\}, \L)$,

\item $s' := (E, \L)$,

\item for $j \in \Ag$
 \[
  \sigma'(j) := 
  \left\{
  \begin{array}{ll}
    \sigma(j) \cup \{m\} & \textrm{if }j\in R(m) \cup S(m) \cup B \\
    \sigma(j)            & \textrm{otherwise}
  \end{array} \right.
 \]
\end{itemize}

We say that the above transition \bfe{processes} the email $m_B$.
This takes place subject to the following conditions
depending on the form of $m$, where $\L = (L_1, \ldots, L_n)$:

\begin{itemize}
\item \textbf{send} $m = s(i,l,G)$. 

We stipulate then that $l \in L_i$ or for some $m' \in \sigma(i)$ we have $l \in FI(m')$. 
In the second case of the second alternative we say below that $m$ \bfe{depends on} $m'$.

\item \textbf{forward}
$m = f(i,l.m',G)$.

We stipulate then that $m' \in \sigma(i)$, and
$l \in L_i$ or for some $m'' \in \sigma(i)$ we have $l \in FI(m'')$.

In the case of the first alternative we say below that $m$ \bfe{depends on} $m'$
and in the case of the second alternative that $m$ \bfe{depends on} $m'$ and $m''$.
\end{itemize}

Given a legal state $s$ an \bfe{email exchange starting in $s$}
is a maximal sequence of transitions starting in the
configuration $<s,\sigma_0>$, where $\sigma_0$ is the empty mailbox.
An email exchange \bfe{properly terminates} if its last configuration is
of the form $<s', \tau>$, where $s' = (\emptyset,\L)$.
The way the atomic transitions are defined clarifies that the communication 
is synchronous.

Note that messages are never deleted from the mailboxes.  Further,
observe that in the above atomic transitions we augment the mailboxes
of the recipients of $m_B$ (including the BCC recipients) by $m$ and
\emph{not} by $m_B$.  So the recipients of $m_B$ only `see' the
message $m$ in their mailboxes. Likewise, we augment the mailbox of the
sender by the message $m$ and \emph{not} by $m_B$.  As a result when
in an email exchange a sender forwards his own email, 
the BCC recipients of the original
email are not shown in the forwarded email. 
This is consistent with the discussion of the emails given
in Subsection \ref{subsec:emails}.

Observe that from the form of a message $m$ in the mailbox $\sigma(i)$
we can infer whether agent $i$ received it by means of a BCC. Namely,
this is the case iff $i \not\in R(m) \cup S(m)$.  (Recall that by
assumption the sets of regular recipients and BCC recipients of an
email are disjoint.)

The following result then clarifies the concept of a legal state.
\begin{theorem}
  The following statements are equivalent:

  \begin{enumerate}[(i)]
  \item $s$ is a legal state,

  \item an email exchange starting in $s$ properly terminates,

  \item all email exchanges starting in $s$ properly terminate.
  \end{enumerate}
\end{theorem}

The equivalence between $(i)$ and $(ii)$ states that
the property of a legal state amounts to the possibility of 
processing all the emails in an orderly (and synchronous) fashion.

\begin{proof}
Suppose $s = (E,\L)$.

$(i) \Rightarrow (ii)$.
Suppose that $s$ is a legal state. So conditions L.\ref{leg1}-L.\ref{leg3} are satisfied
with respect to an spo $\prec$. 
Extend $\prec$ to a linear ordering $\prec_l$ on $E$. 
(Such an extension exists on the account of the result of \cite{Szp30}.)
By the definition of the atomic transitions we can process the emails in $E$ in the order
determined by $\prec_l$. The resulting sequence of transitions forms a properly terminating 
email exchange starting in $s$.

$(ii) \Rightarrow (iii)$.
Let $\xi$ be a properly terminating email exchange starting in $s$ and $\xi'$ another
email exchange starting in $s$.
Let $m_B$ be the first email processed in $\xi$ that is not processed in $\xi'$. 
The final mailbox of $\xi'$ contains the message(s) on which $m$
depends on, since their full versions were processed in $\xi$ before $m_B$
and hence were also processed in $\xi'$.
So $m_B$ can be processed in the final mailbox of $\xi'$, i.e., $\xi'$ is
not a maximal sequence. This is a contradiction.

$(iii) \Rightarrow (ii)$.
Obvious.

$(ii) \Rightarrow (i)$.
Take a properly terminating email exchange $\xi$ starting in $s$. 
Take the following spo $\prec$ on the emails
of $E$: $e_1 \prec e_2$ iff $e_1$ is processed in $\xi$
before $e_2$.  By the definition of the atomic transitions 
conditions L.\ref{leg1}-L.\ref{leg3} are satisfied w.r.t. $\prec$, so $s$ is legal.
\end{proof}

Intuitively, the equivalence between the first two conditions means
that the legality of a state is equivalent to the condition that it is
possible to execute its emails in a `coherent' way.
Each terminating exchange entails a strict partial (in fact linear) ordering
w.r.t. which conditions L.\ref{leg1}-L.\ref{leg3} are satisfied.

\section{Conclusions and future work}
\label{sec:future}

Email is by now one of the most common forms of group communication.
This motivates the study here presented. The language we introduced
allowed us to discuss various fine points of email communication,
notably forwarding and the use of BCC.  The epistemic semantics we
proposed aimed at clarifying the knowledge-theoretic consequences of
this form of communication. Our presentation focused on the issues of
epistemic content of the emails and common knowledge.

This framework also leads to natural questions concerning
axiomatization of the introduced language and the decidability of its 
semantics. Currently we work on 

\begin{itemize}
\item a sound and complete axiomatization of the epistemic language $\Lang$
of Section \ref{sec:epistemic}; at this stage we have such an axiomatization
for the non-epistemic formulas,

\item the problem of decidability of the truth definition given in 
Section \ref{sec:epistemic}; at this stage we have a decidability result
for positive formulas and for formulas without nested epistemic operators,

\item a comparison of the proposed semantics with the one based on 
sequences ('histories') of emails rather than partially ordered sets
of emails.
\end{itemize}

In our framework, as explained in Sections~\ref{sec:epistemic} and
\ref{sec:dis}, communication is synchronous.  This is of course a
simplifying assumption and should be viewed as a first step in
analyzing epistemic reasoning in email exchanges.
We plan to extend our results to the more general framework of
\cite{BM10}, by assuming for each agent a known time bound by which he reads
his emails.

When moving to asynchronous communication we should be aware of the
already mentioned result of \cite{HM90} that common knowledge of
nontrivial facts cannot be achieved. In view of this negative result
various weaker forms of common knowledge were proposed in the
literature.  In particular \cite{NeiTou93} introduced the concept of a
\emph{timestamped common knowledge} and proposed a communication
primitive that achieves it. In turn, \cite{PanTay92} introduced the
notion of a \emph{concurrent common knowledge} defined using Lamport's
causality notion.  These variants could be studied for the case of
email exchanges with asynchronous communication, taking the present
framework as a departure point.

Communication by email suggests other forms of knowledge.  Recently
\cite[Chapter 6]{Sie12} considered \emph{potential knowledge} and
\emph{definitive knowledge} in the context of email exchanges.  When a
message is sent to an agent, that agent acquires potential knowledge
of it. Only when he forwards the message he acquires definitive
knowledge of the message.  The idea is that when a message is sent to
an agent one cannot be sure that he read it.  Only when he forwards it
one can be certain that he did read it. The considered framework is an
adaptation of the one presented here.  The common knowledge is not
considered but a decision procedure is presented for all considered
epistemic formulas.

Another extension worthwhile to study is one in which the agents
communicate richer basic statements than just notes. We already
indicated in Section \ref{sec:bcc} that sending messages containing a
formula $i \blacktriangleleft m$ increases the expressiveness of the
messages from the epistemic point of view.  One could also consider in
our framework sending epistemic formulas, a feature recently studied
in \cite{SvE11} in a setting with a finite number of messages and the
BCC feature absent.

Finally, even though this study was limited to the epistemic aspects
of email exchanges, it is natural to suggest here some desired
features of emails. One is the possibility of forwarding a message in
a provably intact form. This form of forward, used here, is
present in the VM email system integrated into the
\texttt{emacs} editor; in VM forward results in passing the message as
an attachment that cannot be changed.  Another, more pragmatic one and
not considered here, is disabling the reply-all feature for the BCC
recipients so that none of them can by mistake reveal that he was a
BCC recipient.  Yet another one is a feature that would simulate signing
of a reception of a registered letter --- opening such a `registered
email' would automatically trigger an acknowledgment. Such an
acknowledgment would allow one to achieve in a simple way the above
mentioned definitive knowledge.

\subsection*{Acknowledgements}
We thank all three referees for detailed reports that allowed us to improve
the paper in a number of respects.
We acknowledge helpful early discussions with Jan van Eijck and Rohit Parikh
and useful referee comments of the preliminary version that was presented at a workshop.
Yoram Moses helped us to clarify that our results are based on a synchronous communication.

\bibliography{library}

\begin{thebibliography}{10}

\bibitem{AWZ09}
K.~R. Apt, A.~Witzel, and J.~A. Zvesper.
\newblock Common knowledge in interaction structures.
\newblock In {\em Proceedings of TARK XII}, pages 4--13. The ACM Digital
  Library, 2009.

\bibitem{Bab90}
L.~Babai.
\newblock E-mail and the unexpected power of interaction.
\newblock In {\em Fifth Structure in Complexity Theory Conference}, pages
  30--44, 1990.

\bibitem{BM10}
I.~Ben-Zvi and Y.~Moses.
\newblock Beyond {Lamport's} {{\it Happened-Before}}: On the role of time
  bounds in synchronous systems.
\newblock In {\em Proceedings of DISC 2010)}, pages 421--436, 2010.

\bibitem{vBvEK06}
J.~v. Benthem, J.~van Eijck, and B.~Kooi.
\newblock Logics of communication and change.
\newblock {\em Information and Computation}, 204(11):1620--1662, 2006.

\bibitem{chandy_processes_1986}
K.~M. Chandy and J.~Misra.
\newblock How processes learn.
\newblock {\em Distributed Computing}, 1(1):40--52, Mar. 1986.

\bibitem{FHMV_RAK}
R.~Fagin, J.~Halpern, M.~Vardi, and Y.~Moses.
\newblock {\em Reasoning about knowledge}.
\newblock MIT Press, Cambridge, MA, USA, 1995.

\bibitem{HM90}
J.~Halpern and Y.~Moses.
\newblock Knowledge and common knowledge in a distributed environment.
\newblock {\em Journal of the ACM}, 37:549--587, 1990.

\bibitem{Lam78}
L.~Lamport.
\newblock Time, clocks, and the ordering of events in a distributed system.
\newblock {\em Communications of the ACM}, 21(7):558--565, 1978.

\bibitem{EC09}
E.~Lien and P.~C. {\"O}lveczky.
\newblock Formal modeling and analysis of an {IETF} multicast protocol.
\newblock In {\em Seventh IEEE International Conference on Software Engineering
  and Formal Methods, SEFM 2009}, pages 273--282. IEEE Computer Society, 2009.

\bibitem{NeiTou93}
G.~Neiger and S.~Toueg.
\newblock Simulating synchronized clocks and common knowledge in distributed
  systems.
\newblock {\em Journal of the ACM}, 40(2):334--367, 1993.

\bibitem{Pac10}
E.~Pacuit.
\newblock Logics of informational attitudes and informative actions.
\newblock {\em Journal of the Indian Council of Philosophical Research}, 27(2),
  2010.
\newblock 37 pages.

\bibitem{pacpar07}
E.~Pacuit and R.~Parikh.
\newblock Reasoning about communication graphs.
\newblock {\em Interactive Logic. Proceedings of the 7th Augustus de Morgan
  Workshop}, pages 135--157, 2007.

\bibitem{PanTay92}
P.~Panangaden and K.~Taylor.
\newblock Concurrent common knowledge: Defining agreement for asynchronous
  systems.
\newblock {\em Distributed Computing}, 6:73--93, 1992.

\bibitem{PR03}
R.~Parikh and R.~Ramanujam.
\newblock A knowledge based semantics of messages.
\newblock {\em Journal of Logic, Language and Information}, 12(4):453--467,
  2003.

\bibitem{Plo82}
G.~D. Plotkin.
\newblock An operational semantics for {CSP}.
\newblock In D.~Bj{\o}rner, editor, {\em Formal Description of Programming
  Concepts II}, pages 199--225, Amsterdam, 1982. North-Holland.

\bibitem{Sie12}
F.~Sietsma.
\newblock An applicable logic of emails and knowledge.
\newblock Manuscript, 2012.

\bibitem{SvE11}
F.~Sietsma and J.~van Eijck.
\newblock Message passing in a dynamic epistemic logic setting.
\newblock In {\em Proceedings of TARK XIII}, pages 212--220. The ACM Digital
  Library, 2011.

\bibitem{Szp30}
E.~Szpilrajn.
\newblock Sur l'extension de l'ordre partiel.
\newblock {\em Fundamenta Mathematicae}, 16:386--389, 1930.

\bibitem{WSE10}
Y.~Wang, F.~Sietsma, and J.~van Eijck.
\newblock Logic of information flow on communication channels.
\newblock In {\em Proceedings of AAMAS-10}, pages 1447--1448, 2010.

\end{thebibliography}
\bibliographystyle{abbrv}

\end{document}